\documentclass[journal,twoside,web]{ieeecolor}
\usepackage{generic}
\usepackage{cite}

\usepackage{graphicx,algorithm,algpseudocode,amsthm,amsmath,amsfonts,verbatim,mathtools,enumerate,bm,hyperref,carshapes,aircraftshapes,hyperref}
\usepackage[style=base]{caption}
\usetikzlibrary{calc,patterns,positioning}
\usepackage{subcaption}
\allowdisplaybreaks

\newcommand{\diag}{\mathop{\rm diag}}
\newcommand{\blkdiag}{\mathop{\rm blkdiag}}
\newcommand{\tr}{\mathop{\rm tr}}

\newcommand{\argmin}{\mathop{\rm argmin}}

\newcommand{\norm}[1]{\left\lVert#1\right\rVert}
\newcommand{\mnorm}[1]{{\left\vert\kern-0.25ex\left\vert\kern-0.25ex\left\vert #1 
    \right\vert\kern-0.25ex\right\vert\kern-0.25ex\right\vert}}

\newtheorem{definition}{Definition} 

\newtheorem{lemma}{Lemma}

\newtheorem{remark}{Remark}
\newtheorem{proposition}{Proposition}
\newtheorem{assumption}{Assumption}

\newcommand{\ie}{{\it i.e.}}

\definecolor{new}{rgb}{0.35, 0, 0.4}

\definecolor{car1}{rgb}{0, 0.28, 0.67}
\definecolor{car2}{rgb}{0, 0.42, 0.24}
\definecolor{dgreen}{rgb}{0, 0.5, 0}

\definecolor{blue}{rgb}{0.1, 0.1, 1}

\begin{document}

\title{Inverse Matrix Games with Unique Quantal Response Equilibrium}

\author{Yue~Yu, Jonathan~Salfity, David~Fridovich-Keil, and Ufuk~Topcu 
\thanks{Y. Yu and U. Topcu are with the Oden Institute for Computational Engineering and Sciences, The University of Texas at Austin, TX, 78712, USA (emails:  yueyu@utexas.edu,\,utopcu@utexas.edu). J. Salfity is with the Department of Mechanical Engineering, The University of Texas at Austin, TX, 78712, USA (email:  j.salfity@utexas.edu). D. Fridovich-Keil is with the Department of Aerospace Engineering and Engineering Mechanics, The University of Texas at Austin, TX, 78712, USA (email: dfk@utexas.edu). 
}
}

\maketitle
\thispagestyle{empty}
\pagestyle{empty}

\begin{abstract}
In an inverse game problem, one needs to infer the cost function of the players in a game such that a desired joint strategy is a Nash equilibrium. We study the inverse game problem for a class of multiplayer matrix games, where the cost perceived by each player is corrupted by random noise. We provide sufficient conditions for the players' quantal response equilibrium---a generalization of the Nash equilibrium to games with perception noise---to be unique. We develop efficient optimization algorithms for inferring the cost matrix based on semidefinite programs and bilevel optimization. We demonstrate the application of these methods in encouraging collision avoidance and fair resource allocation.
\end{abstract}


\section{Introduction}
\label{sec: introduction}

In a multiplayer game, each player tries to find the strategies with the minimum cost, where the cost of each strategy depends on the other players' strategies. The Nash equilibrium is a set of strategies where no player can benefit from unilaterally changing strategies. The Nash equilibrium generalizes minimax equilibrium in two-player zero-sum games \cite{morgenstern1953theory} to multiplayer general-sum games \cite{nash1950equilibrium,nash1951non}. 

Given a joint strategy of the players in a game, the inverse game problem requires inferring the cost function such that the given joint strategy is indeed a Nash equilibrium. The inferred cost function can either rationalize observed player behavior \cite{molloy2022inverse,bertsimas2015data,ling2018game} or provide incentives to encourage desired behavior \cite{bacsar1984affine,nisan2015algorithmic}. There have been many results on inverse games in different contexts, including specific games, such as matching \cite{kalyanaraman2008complexity}, network formation \cite{kalyanaraman2009complexity}, and auction \cite{nekipelov2015econometrics}; and generic classes of games, such as succinct games \cite{kuleshov2015inverse}, dynamic games \cite{molloy2022inverse},  and general convex games whose Nash equilibria are characterized by variational inequalities \cite{bertsimas2015data}.

The existing results on inverse games have the following limitations. First, the existing results do not guarantee a unique Nash equilibrium. Such nonuniqueness makes the players' behavior less predictable since there is ambiguity in which equilibrium the players will choose. It also complicates the players' decision-making, since the players need to align their choice of equilibrium with the other players' \cite{peters2020inference}; see Fig.~\ref{fig: nonunique} for an illustrative example. Second, the existing results assume each player has perfect perceptions of the cost of each action. Such an assumption is not reasonable to model human behavior where the players have bounded rationality and imperfect cost estimation \cite{mckelvey1995quantal,mckelvey1998quantal}. Recent work addressed these limitations in the context of two-player zero-sum games \cite{ling2018game}. However, for multi-player general-sum games, treatments to these limitations are, to our best knowledge, still missing. 

\begin{figure}[!t]
		\centering
		\begin{tikzpicture}[scale=0.7]
		\coordinate (O) at (0, 0);
		\coordinate (A) at ($(O)+(-3, 0)$);
		\coordinate (B) at ($(O)+(3, 0)$);
		
 		\coordinate (C) at ($(A)+(0.8, -0.8)$);
 		\coordinate (D) at ($(B)+(-0.8, -0.8)$);
		
		\node [sedan top,body color=car1!60,window color=black!80,minimum width=1.5cm] at (A) {};
		\node [sedan top,body color=car2!60,window color=black!80,minimum width=1.5cm,rotate=180] at (B) {};
		
		\draw[car2, -latex, very thick] ($(O)+(10:1.5)$) arc (10:85:1.5) node[midway, above, sloped] {\color{car2} \scriptsize strategy 1};
		\draw[car2, dotted, -latex, very thick] ($(O)+(-10:1.5)$) arc (-10:-85:1.5) node[midway, below, sloped] {\color{car2} \scriptsize strategy 2};
		
		\draw[car1, dotted, -latex, very thick] ($(O)+(170:1.5)$) arc (170:95:1.5) node[midway, above, sloped] {\color{car1} \scriptsize strategy 2};
		\draw[car1, -latex, very thick] ($(O)+(190:1.5)$) arc (190:265:1.5) node[midway, below, sloped] {\color{car1} \scriptsize strategy 1};
		
 		\node[label={left:{\color{car1} \scriptsize player 1}}] at (C) {};
 		\node[label={right:{\color{car2} \scriptsize player 2}}] at (D) {};
		
\end{tikzpicture}
		\caption{A two-player game with two Nash equilibrium: joint strategy \((1, 1)\) and \((2, 2)\). The players must align their choice of equilibrium to avoid collisions.}\label{fig: nonunique}
\end{figure}
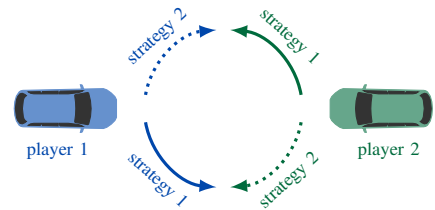

We study the inverse game problem for a class of multi-player general-sum matrix games, where each player's strategy is a probability distribution over a finite number of discrete actions, the cost of a strategy is characterized by a matrix, and the cost perceived by each player is corrupted by random noise. Our contributions are as follows.

First, we provide sufficient conditions on the cost matrix for the uniqueness of the quantal response equilibrium---which is a generalization of the Nash equilibrium when the cost perceived by each player is corrupted by noise \cite{mckelvey1995quantal,mckelvey1998quantal}---and show that one can efficiently compute this unique equilibrium by solving a nonlinear least-squares problem. Second, we develop two numerical methods–one based on semidefinite programs, the other based on bilevel optimization–that infer the cost matrices which optimize the unique quantal response equilibrium with respect to a performance function. The latter extends the implicit differentiation methods in \cite{ling2018game} from convex-concave saddle point problems to nonconvex equilibrium problems. Finally, we demonstrate the application of these methods in encouraging collision avoidance and fair resource allocation.

Our results are particularly useful for designing incentives that motivate desired behavior when the players' cost function is (partially) known. They also have potential applications in general inverse multiplayer games.

\paragraph*{Notation} We let \(\mathbb{R}\), \(\mathbb{R}_+\), \(\mathbb{R}_{++}\), and \(\mathbb{N}\) denote the set of real, nonnegative real, positive real, and positive integer numbers, respectively. Given \(m, n\in\mathbb{N}\), we let \(\mathbb{R}^n\) and \(\mathbb{R}^{m\times n}\) denote the set of \(n\)-dimensional real vectors and \(m\times n\) real matrices; we let \(\mathbf{1}_n\) and \(I_n\) denote the \(n\)-dimensional vector of all 1's and the \(n\times n\) identity matrix, respectively.  Given positive integer \(n\in\mathbb{N}\), we let \([n]\coloneqq \{1, 2, \ldots, n\}\) denote the set of positive integers less or equal to \(n\). Given \(x\in\mathbb{R}^n\) and \(k\in[n]\), we let \([x]_k\) denote the \(k\)-the element of vector \(x\), and \(\norm{x}\) denote the \(\ell_2\)-norm of \(x\).  Given a square real matrix \(A\in\mathbb{R}^{n\times n}\), we let \(A^\top\), \(A^{-1}\), and \(A^{-\top}\) denote the transpose, the inverse, and the transpose of the inverse of matrix \(A\), respectively; we say \(A\succeq 0\) and \(A\succ 0\) if \(A\) is symmetric positive semidefinite and symmetric positive definite, respectively; we let \(\norm{A}_F\) denote the Frobenius norm of matrix \(A\). We let \(\blkdiag(A_1, \ldots, A_k)\) denote the block diagonal matrix whose diagonal blocks are \(A_1, \ldots, A_k\in\mathbb{R}^{m\times m}\). Given continuously differentiable functions \(f:\mathbb{R}^n\to\mathbb{R}\) and \(G:\mathbb{R}^n\to\mathbb{R}^m\), we let \(\nabla_x f(x)\in\mathbb{R}^n\) denote the gradient of \(f\) evaluated at \(x\in\mathbb{R}^n\); the \(k\)-th element of \(\nabla_x f(x)\) is \(\frac{\partial f(x)}{\partial [x]_k}\). Furthermore,  we let \(\partial_x G(x)\in\mathbb{R}^{m\times n}\) denote the Jacobian of function \(G\) evaluated at \(x\in\mathbb{R}^n\); the \(ij\)-th element of \(\partial_x G(x)\) is \(\frac{\partial [G(x)]_i}{\partial [x]_j}\). 
\section{Quantal response equilibrium in matrix games}
\label{sec: entropy matrix games}

We introduce our theoretical model, a multiplayer matrix game where the cost perceived by each player is corrupted by stochastic noise.

\subsection{Multiplayer matrix games with perception noise}
\label{subsec: matrix game}

We consider a game with \(n\in\mathbb{N}\) players. Each player \(i\in[n]\) has \(m_i\in\mathbb{N}\) actions. We let \(m\coloneqq \sum_{i=1}^n m_i\) denote the total number of actions of all players. Player \(i\)'s strategy is an \(m_i\)-dimensional probability distribution over all possible actions, denoted by \(x_i\in\Delta_i\), where
\begin{equation}
    \Delta_i \coloneqq \{y\in\mathbb{R}^{m_i}| y^\top \mathbf{1}_{m_i}=1, y\geq 0\}.
\end{equation}

Each player's optimal strategy is one that minimizes the perceived cost, which is jointly determined by the strategies of all players and a stochastic perception error. In particular, let
\(b_i\in\mathbb{R}^{m_i}\) and \(C_{ij}\in\mathbb{R}^{m_i\times m_j}\) with \(C_{ii}=C_{ii}^\top\succeq 0\) for all \(i, j\in[n]\) be cost parameters. Then the cost of action \(k\) perceived by player \(i\) is given by
\begin{equation}
    [\textstyle b_i+\sum_{j=1}^n C_{ij}x_j]_k+\xi_{ik},
\end{equation}
where \(\xi_{ik}\) is a random variable that captures perception error in player \(i\)'s decision-making. 

If each \(\xi_{ik}\) is independently sampled from the Gumbel distribution with mean \(\gamma\lambda\) for some \(\lambda\in\mathbb{R}_{++}\) where \(\gamma\) is Euler's constant, then the optimal strategy for player \(i\) is
\begin{equation}\label{eqn: softmax}
\begin{aligned}
    \textstyle x_i&=\textstyle f_i\big(-\frac{1}{\lambda}(b_i+\sum_{j=1}^n C_{ij}x_j)\big)
\end{aligned}
\end{equation}
where \(f_i(z)\coloneqq\frac{1}{\mathbf{1}_{m_i}^\top \exp(z)}\exp(z)\) for all \(z\in\mathbb{R}^{m_i}\), and \(\exp(z)\in\mathbb{R}_{++}^{m_i}\) is the elementwise exponential of vector \(z\). The strategy in \eqref{eqn: softmax}, known as the \emph{logit quantal response}, models the bounded rationality in decision-making, and has been effective in consumer choice problems \cite{mcfadden1973conditional}; see \cite{mckelvey1995quantal,mckelvey1998quantal} for a detailed discussion. We define the concept of \emph{logit quantal response equilibrium} formally as follows.

\begin{definition}\label{def: entropy Nash}
A joint strategy \(x\coloneqq \begin{bmatrix}
x_1^\top & x_2^\top & \ldots & x_n^\top
\end{bmatrix}^\top\) is a \emph{quantal response equilibrium} if \eqref{eqn: softmax} holds for all \(i\in[n]\).
\end{definition}

The following lemma provides an optimization-based characterization of the equilibrium in Definition~\ref{def: entropy Nash}. 

\begin{lemma}\label{lem: softmax}
If \(\lambda>0\) and \(C_{ii}=C_{ii}^\top\succeq 0\) for all \(i\in[n]\), then \eqref{eqn: softmax} holds if and only if 
\begin{equation}\label{opt: entropy}
x_i\in\underset{y\in\Delta_i}{\argmin} \enskip \textstyle \big(b_i+\frac{1}{2}C_{ii} y+\sum_{j\neq i} C_{ij}x_j\big)^\top y+\lambda y^\top \ln(y),
\end{equation}
where \(\ln(y)\) denotes the elementwise logarithm of \(y\).
\end{lemma}
\begin{proof}

Since \(C_{ii}=C_{ii}^\top\succeq 0\) and the set \(\Delta_i\) is nonempty, the condition in \eqref{opt: entropy} holds if and only if the following Karush-Kuhn-Tucker conditions hold for some \(v_i\in\mathbb{R}\) and \(u_i\in\mathbb{R}^{m_i}\)
\begin{equation}\label{eqn: entropy kkt}
\begin{aligned}
     &\textstyle b_i+\sum_{j=1}^n C_{ij}x_j+(\lambda+v_i) \mathbf{1}_{m_i}+\lambda\ln(x_i)-u_i=0_{m_i},\\
     &x_i^\top\mathbf{1}_{m_i}=1,\enskip x_i\geq 0, \enskip u_i\geq 0, \enskip u_i^\top x_i=0.
\end{aligned}
\end{equation}Furthermore, since the logarithm function is only defined for strictly positive numbers, we know that \eqref{opt: entropy} implies that \(x_i\) is elementwise positive. Combining this fact with \eqref{eqn: entropy kkt} gives \(u_i=0_{m_i}\). Finally, one can directly verify the equivalence between \eqref{eqn: softmax} and \eqref{eqn: entropy kkt} when \(u_i=0_{m_i}\).

\end{proof}
\begin{remark}
    If \(n=2\), \(C_{ii}=0_{m_i\times m_i}\) for \(i=1, 2\), and \(C_{12}=-C_{21}\), then Definition~1 reduces to  the quatal response equlibrium in two-player zero-sum games \cite{ling2018game}. 
\end{remark}

\subsection{Computing the quantal response equilibrium via nonlinear least-squares}

We can compute the quantal response equilibrium in Definition~\ref{def: entropy Nash} by solving the following \emph{nonlinear least-squares problem}:
\begin{equation}\label{opt: nonlinear ls}
    \begin{array}{ll}
        \underset{x}{\mbox{minimize}} &  \sum_{i=1}^n \norm{x_i-  f_i\big(-\frac{1}{\lambda}(b_i+\sum_{j=1}^n C_{ij}x_j)\big)}^2
    \end{array}
\end{equation}
where function \(f_i\) is given by \eqref{eqn: softmax}. If the optimal value of the objective function in optimization~\eqref{opt: nonlinear ls} is zero, then the corresponding solution \(x\) indeed satisfies \eqref{eqn: softmax} for all \(i\in[n]\).  

However, the question remains whether 
optimization~\eqref{opt: nonlinear ls} has an optimal value of zero, or whether it has a unique solution. We will answer these questions next. 

Throughout we will also use the following notation:
\begin{equation}\label{eqn: game parameter}
b \coloneqq \begin{bsmallmatrix}
b_1\\
b_2\\
\vdots\\
b_n
\end{bsmallmatrix}, \enskip
    C\coloneqq\begin{bsmallmatrix}
    C_{11} & C_{12}  & \hdots & C_{1n}\\
    C_{21} & C_{22}  & \hdots & C_{2n}\\
    \vdots & \vdots & \ddots & \vdots \\
    C_{n1} & C_{n2} & \hdots  & C_{nn}
    \end{bsmallmatrix}.
\end{equation}

We make the following assumptions on optimization~\eqref{opt: nonlinear ls}.
\begin{assumption}\label{asp: main}
    \(\lambda>0\), \(C+C^\top\succeq 0\), \(C_{ii}=C_{ii}^\top\) for all \(i\in[n]\).
\end{assumption}

The following proposition provides sufficient conditions under which the quantal response equilibrium in Definition~\ref{def: entropy Nash} exists and is unique.

\begin{proposition}\label{prop: unique}
If Assumption~\ref{asp: main} holds, then there exists a unique \(x= \begin{bmatrix}
x_1^\top & x_2^\top &\ldots & x_n^\top
\end{bmatrix}^\top\in\mathbb{R}^m_{++}\) such that \eqref{eqn: softmax} holds for all \(i\in[n]\). 
\end{proposition}
\begin{proof}

We first prove the existence. Since \(C+C^\top\succ 0\), all of its principle submatrices are also positive semidefinite. Hence \(C_{ii}\succeq 0\) for all \(i\in[n]\). In addition, since \(y^\top\ln(y)\) is a convex function of \(y\), we conclude that any \(x\) that satisfies \eqref{opt: entropy} is the Nash equilibrium of a concave game, which always exists, due to \cite[Thm. 1]{rosen1965existence}.

Next, we prove the uniqueness of \(x\). The logarithm function in \eqref{opt: entropy} ensures that \(x\) is elementwise positive. By combining this fact together with Assumption~\ref{asp: main}, we can show that \(C+C^\top+\lambda \diag(x)^{-1}\)  is positive definite. Hence any \(x\) that satisfies \eqref{opt: entropy} is the Nash equilibrium of a diagonally strict concave game, which is unique \cite[Thm. 6]{rosen1965existence}.

The rest of the proof follows from the equivalence between \eqref{eqn: softmax} and \eqref{opt: entropy}, due to Lemma~\ref{lem: softmax}.

\end{proof}

In practice, different cost functions can induce the same equilibrium, even those violating Assumption~\ref{asp: main} . The cost functions satisfying Assumption~\ref{asp: main}, however, eliminate the ambiguity in the quantal response equilibrium, as shown by Proposition~\ref{prop: unique}.


\section{Numerical methods for inverse matrix games}
\label{sec: infer}

We now consider the following inverse game problem: given a desired joint strategy \(x\), how can one infer the cost matrix \(C\) that makes \(x\) the unique equilibrium in Definition~\ref{def: entropy Nash}. Here we only consider the inferring of the matrix \(C\). We note that one can seamlessly generalize the results in this section to the inference of vector \(b\).

In the following, we will introduce two different approaches for the aforementioned inverse matrix game: one based on semidefinite programs, the other based on the projected gradient method for bilevel optimization.

\subsection{Semidefinite program approach} 
\label{subsec: sdp}
We first consider the case where the desired equilibrium is a pure joint strategy, where each player \(i\) has a preferred action \(i^\star\in[m_i]\). In particular, suppose there exists \(x^\star\in\mathbb{R}^m\) and \(i^\star\in[m_i]\) for all \(i\in[n]\) such that \([x_i^\star]_k\) equals \(1\) if \(k=i^\star\), and \(0\) otherwise.

In this case, perhaps the most direct way to ensure \(x^\star\) is an equilibrium is to simply make sure that the cost of action \(i^\star\) is sufficiently lower than any alternatives for player \(i\). By combining these constraints together with the results in Proposition~\ref{prop: unique}, we obtain the following semidefinite program: 
\begin{equation}\label{opt: sdp}
    \begin{array}{ll}
        \underset{C}{\mbox{minimize}} & \frac{1}{2}\norm{C}_F^2 \\
        \mbox{subject to} & C+C^\top \succeq 0, \enskip C_{ii}=C_{ii}^\top, \enskip i\in[n],\\
        & [b_i]_{i^\star} +\sum_{j=1}^nC_{i^\star j^\star}+\varepsilon\leq [b_i]_{k} +\sum_{j=1}^nC_{kj^\star},\\
        &\forall k\in[m_i]\setminus \{i^\star\}, \enskip i\in[n].
    \end{array}
\end{equation}
where the objective function penalizes large values of the elements in matrix \(C\), and \(\varepsilon\in\mathbb{R}_+\) is a tuning parameter that separates the cost of the best action from the cost of the second best action. Intuitively, as the value of \(\varepsilon\) increases, the quantal response equilibrium in Definition~\ref{def: entropy Nash} is more likely to take a pure form.

The drawback of optimization~\eqref{opt: sdp} is that it only applies to the case where the desired equilibrium is known and close to be deterministic. If the desired equilibrium is mixed, \ie, each player has a preferred probability distribution over all actions rather than one single preferred action, then the semidefinite program is no longer useful. 

\subsection{Bilevel optimization approach}
\label{subsec: bilevel}
We now consider the case where the desired equilibrium is described by a \emph{performance function}, rather than explicitly as a desired joint strategy. In particular, we consider the following continuously differentiable function, denoted by \(\psi: \mathbb{R}^m\to\mathbb{R}\), that evaluates the quality of a joint strategy.
For example, if \(x^\star = \begin{bmatrix}
(x_1^\star)^\top & (x_2^\star)^\top & \cdots & (x_n^\star)^\top  
\end{bmatrix}^\top\) is the desired equilibrium, then a possible choice of function \(\psi\) is as follows:
\begin{equation}\label{eqn: KL div}
   \textstyle  \psi(x)=D_{KL}(x, x^\star)\coloneqq \sum_{i=1}^n x_i^\top (\ln(x_i)-\ln(x_i^\star)).
\end{equation}
The above choice of function \(\psi(x)\) measures the sum of the Kullback–Leibler (KL) divergence between each player's strategy and the corresponding desired strategy.

In order to compute the value of matrix \(C\) such that the equilibrium in Definition~\ref{def: entropy Nash} is unique and minimizes the value of performance function \(\psi(x)\), we introduce the following \emph{bilevel optimization problem}: 
\begin{equation}\label{opt: bilevel}
    \begin{array}{ll}
        \underset{x, C}{\mbox{minimize}} & \psi(x)\\
       \mbox{subject to} & C+C^\top\succeq 0,\enskip C_{ii}=C_{ii}^\top, \enskip i\in[n],\\
       &\norm{C}_F\leq \rho,\enskip \text{\(x\) is optimal for optimization~\eqref{opt: nonlinear ls}. } 
    \end{array}
\end{equation}
Here \(\rho\in\mathbb{R}_+\) is a  tuning parameter that controls the maximum allowed Frobenius norm of matrix \(C\). Intuitively, the larger the value of \(\rho\), the more choices of matrix \(C\) from which we can choose, and the more likely we can achieve a lower value of function \(\psi(x)\). 

The drawback of optimization~\eqref{opt: bilevel} is that, unlike the semidefinite program in \eqref{opt: sdp}, it is nonconvex and, as a result, one can only hope to obtain a locally optimal solution. Next, we will discuss how to compute a locally optimal solution efficiently using the projected gradient method.

\subsubsection{Differentiating through the equilibrium condition}
\label{subsec: differentiation}
The key to solve bilevel optimization~\eqref{opt: bilevel} is to compute the gradient of \(\psi(x)\) with respect to matrix \(C\). In particular, we let \(\nabla_C\psi(x)\in\mathbb{R}^{m\times m}\) be the matrix whose \(pq\)-th element, denoted by \([\nabla_C\psi(x)]_{pq}\), is given by
\begin{equation}\label{eqn: C grad}
   \textstyle [\nabla_C\psi(x)]_{pq}\coloneqq \frac{\partial \psi(x)}{\partial [C]_{pq}}
\end{equation}
for all \(p, q\in[m]\). Since function \(\psi\) is continuously differentiable, the difficulty in evaluating \(\nabla_C\psi(x)\) is to compute the Jacobian of the equilibrium \(x\) with respect to matrix \(C\). To this end, we introduce the following notation:
\begin{subequations}\label{eqn: ufu}
\begin{align}
    u\coloneqq &\textstyle -\frac{1}{\lambda}(b+Cx)\\
    f(u)\coloneqq &\begin{bmatrix}
    f_1(u_1)^\top & f_2(u_2)^\top & \cdots f_n(u_n)^\top 
    \end{bmatrix}^\top
\end{align}
\end{subequations}
where \(u_i\in\mathbb{R}^{m_i}\) for all \(i\in[n]\), and \(f_i\) is given by Lemma~\ref{lem: softmax}. The following result provides a formula to compute \(\nabla_C\psi(x)\) using the \emph{implicit function theorem} \cite{dontchev2014implicit}.

\begin{proposition}\label{prop: gradinet}
Suppose \(C+C^\top\succeq 0\) and \(\lambda>0\). Let \(x\coloneqq \begin{bmatrix}
x_1^\top & x_2^\top &\ldots & x_n^\top
\end{bmatrix}^\top\) be such that \eqref{eqn: softmax} holds for all \(i\in[n]\),
\(\psi:\mathbb{R}^{m}\to\mathbb{R}\) be a continuously differentiable function, \(u\) and \(f(u)\) given by \eqref{eqn: ufu}. If \(I_m+\frac{1}{\lambda}\partial_uf(u) C\) is nonsingular, then
\begin{equation*}
    \textstyle \nabla_C\psi(x) =   -\frac{1}{\lambda}\partial_uf(u)^\top (I_m+\frac{1}{\lambda}\partial_uf(u) C)^{-\top} \nabla_x \psi(x) x^\top.
\end{equation*}
\end{proposition}
\begin{proof}
Let \(F(x, C)\coloneqq x-f(u)\) and \(C_q\) denote the \(q\)-th column of matrix \(C\). Proposition~\eqref{prop: unique} implies \(x\) is the unique vector that satisfies \(F(x, C)=0_m\). Since \(f\) is a continuously differentiable function, the implicit function theorem \cite[Thm. 1B.1]{dontchev2014implicit} implies the following: if \(\partial_x F(x, C)\) is nonsingular, then \(\frac{\partial x}{\partial C_q} = -(\partial_x F(x, C))^{-1}\partial_{C_q}F(x, C)\). Using the chain rule we can show \(\partial_x F(x, C)=I_m+\frac{1}{\lambda}\partial_uf(u)C\) and \(\partial_{C_q}F(x, C)=\frac{1}{\lambda}[x]_q\partial_u f(u)\). The rest of the proof is due to the chain rule and the definition of \(\nabla_C\psi(x)\) in\eqref{eqn: C grad}. 
\end{proof}

The gradient formula in Proposition~\eqref{prop: gradinet} requires computing matrix inverse, which can be numerically unstable, In practice, we use the following formula:
\begin{equation}\label{eqn: grad}
     \textstyle \nabla_C\psi(x) =   -\frac{1}{\lambda}\partial_uf(u)^\top ((I_m+\frac{1}{\lambda}\partial_uf(u) C)^\dagger)^\top \nabla_x \psi(x) x^\top,
\end{equation}
where \(^\dagger\) denotes the Moore–Penrose pseudoinverse. Note that if \(I_m+\frac{1}{\lambda}\partial_u f(u) C\) is nonsingular, then Proposition~\ref{prop: gradinet} implies \(\hat{\nabla}_C\psi(x)=\nabla_C \psi(x)\); otherwise, the value of \(\hat{\nabla}\psi(x)\) provides only an \emph{approximation} of  \(\nabla_C \psi(x)\).  

\subsubsection{Approximate projected gradient method}
\label{subsec: proj grad}

Equipped with Proposition~\ref{prop: gradinet} and the projection formula in \eqref{eqn: grad}, we are now ready to introduce the approximate projected gradient method for bilevel optimization~\eqref{opt: bilevel}. To this end, we define the set \(\mathbb{D}\subset\mathbb{R}^{m\times m}\):
\begin{equation}\label{eqn: set D}
\begin{aligned}
     \mathbb{D}\coloneqq \{C| &C+C^\top\succeq 0, \enskip \norm{C}_F\leq \rho,\enskip C_{ii}=C_{ii}^\top, \enskip i\in[n]\}.
\end{aligned}
\end{equation}

We summarize the approximate projected gradient method
in Algorithm~\ref{alg: proj grad}, where the projection map \(\Pi_{\mathbb{D}}:\mathbb{R}^{m\times m}\to\mathbb{R}^{m\times m}\) is given by
\begin{equation}\label{eqn: proj D}
    \Pi_{\mathbb{D}}(C) = \underset{X\in\mathbb{D}}{\argmin}\enskip \norm{X-C}_F
\end{equation}
for all \(C\in\mathbb{R}^{m\times m}\). At each iteration, this method first solve the nonlinear least-squares problem in \eqref{opt: nonlinear ls}, then update matrix \(C\) using the approximate gradient in \eqref{eqn: grad}.

\begin{algorithm}[!ht]
\caption{Approximate projected gradient method. }
\begin{algorithmic}[1]
\Require Function \(\psi:\mathbb{R}^{m}\to\mathbb{R}\), vector \(b\in\mathbb{R}^m\), scalar weight \(\lambda\in\mathbb{R}_{++}\), step size \(\alpha\in\mathbb{R}_{++}\), stopping tolerance \(\epsilon\).
\State Initialize \(C=0_{m\times m}\), \(C^+=2\epsilon I_m\)
\While{\(\norm{C^+-C}_F>\epsilon\)}
\State \(C\gets C^+\)
\State Solve optimization~\eqref{opt: nonlinear ls} for \(x\).
\State \(C^+\gets \Pi_{\mathbb{D}}(C-\alpha \hat{\nabla}_C\psi(x))\)
\EndWhile
\Ensure Equilibrium \(x\) and cost matrix \(C\).
\end{algorithmic}
\label{alg: proj grad}
\end{algorithm}

A key step in Algorithm~\ref{alg: proj grad} is to compute the projection in \eqref{eqn: proj D}. The following lemma provides the explicit computational formula for computing this projection via eigenvalue decomposition and matrix normalization. 

\begin{lemma}
Let set \(\mathbb{D}\) be given by \eqref{eqn: set D} and matrix \(C\in\mathbb{R}^{m\times m}\) be partitioned as in \eqref{eqn: game parameter}, where \(C_{ij}\in\mathbb{R}^{m_i\times m_j}\), for all \(I, j\in[n]\). Let \(B=C-\blkdiag(C_{11}, \ldots, C_{nn}^\top)\), \(U\in\mathbb{R}^{m\times m}\), and \(s\in\mathbb{R}^m\) be such that \(U\diag(s)U^\top=\frac{1}{2}(C+C^\top)\). Let \(\) Then
\begin{equation}
\Pi_{\mathbb{D}}(C)=\frac{\rho}{\max\{\rho, \norm{A}_F\}} A,\label{eqn: proj set}
\end{equation}
where \(A\coloneqq \frac{1}{2}(B-B^\top)+U\diag(\max(s, 0))U^\top\).
\end{lemma}
\begin{proof}

Let \(\tr(M)\) denote the trace of matrix \(M\). First, let \(\mathbb{K}_1\coloneqq\{M\in\mathbb{R}^{n\times n}| M=M^\top\succeq 0\}\), \(\mathbb{K}_2\coloneqq\{M\in\mathbb{R}^{n\times n}| M=-M^\top, \enskip M_{ii}=M_{ii}^\top, \enskip \forall i\in[n]\}\) where \(M_{ii}\subset\mathbb{R}^{m_i\times m_i}\) denotes the \(i\)-th principal submatrix of \(M\), and \(\mathbb{B}\coloneqq \{M\in\mathbb{R}^{n\times n}|\norm{M}_F\leq \rho\}\). Then one can verify that \(\mathbb{D}=(\mathbb{K}_1+\mathbb{K}_2)\cap\mathbb{B}\) where \(+\) denotes the direct sum, and that \(\mathbb{K}_1+\mathbb{K}_2\) is a closed convex cones. By using the results on the projection onto the intersection of a ball and a closed convex cone \cite[Thm. 7.1]{bauschke2018projecting}, we can show that \(\Pi_{\mathbb{D}}(C)=\Pi_{\mathbb{B}}(\Pi_{\mathbb{K}_1+\mathbb{K}_2}(C))\).

Second, observe that \(C\in \mathbb{K}_1+\mathbb{K}_2\) if and only if \(\tilde{C}\coloneqq (\frac{1}{2}(C+C^\top), \frac{1}{2}(C-C^\top))\in \mathbb{K}_1\times \mathbb{K}_2\), where \(\times\) denotes the Cartesian product. By using the results on the projection onto the Cartesian product of sets \cite[Prop. 29.4]{bauschke2017convex}, we can show that \(\Pi_{\mathbb{K}_1\times \mathbb{K}_2} (\tilde{C})=(\Pi_{\mathbb{K}_1}(\frac{1}{2}(C+C^\top)), \Pi_{\mathbb{K}_2}(\frac{1}{2}(C-C^\top)))\). Hence \(\Pi_{\mathbb{K}_1+\mathbb{K}_2}(C)=\Pi_{\mathbb{K}_1}(\frac{1}{2}(C+C^\top))+\Pi_{\mathbb{K}_2}(\frac{1}{2}(C-C^\top))\). The rest of the proof is a direct application of the formulas for projecting onto norm balls \cite[Ex. 3.18]{bauschke2017convex} and projecting a symmetric matrix onto the positive semidefinite cone \cite[Ex. 29.32]{bauschke2017convex}.
\end{proof}

\section{Numerical examples}
\label{sec: numerical}
We demonstrate the application of the numerical methods in Section~\ref{sec: infer} using two examples. In these examples, we assume that vector \(b\) is known and \(C\) is a zero matrix before our design process. We aim to design a nonzero matrix \(C\)---which can be interpreted as subsidies and tolls---that encourages the desired behavior. We note that if the value of either vector \(b\) or matrix \(C\) is unknown before design, one can first infer the values of \(b\) and \(C\) using the approaches in Section~\ref{sec: infer}, then proceed with the results in this section.

Throughout, we compute the entropy regularized equilibrium in Definition~\ref{def: entropy Nash} by solving optimization~\eqref{opt: nonlinear ls} using the Gauss-Newton method with line search \cite[Sec. 10.3]{nocedal1999numerical}, where we terminated the algorithm when the objective function value in \eqref{opt: nonlinear ls} is less than \(10^{-10}\). We note that, depending on the problem settings, other solution algorithms may have better performance. We refer the interested readers to \cite[Ch. 10]{nocedal1999numerical} for a detailed discussion on nonlinear least-squares.

\subsection{Encouraging collision avoidance}

We consider four ground rovers placed in a two-dimensional environment, at coordinate \((0, 1)\), \((0, -1)\), \((1, 0)\), and \((-1, 0)\), respectively. Each rover wants to reach the corresponding target position with coordinates \((0, -1)\), \((0, 1)\), \((-1, 0)\) and \((1, 0)\), respectively. Each rover can choose one of three candidate paths that connect its initial position to its target position: a beeline path of length \(2\); two semicircle paths, each of approximate length \(\pi\), and one in the clockwise direction, the other one in the counterclockwise direction. We assume all rovers move at the same speed and start at the same time. 

We model the decision-making of each rover using the entropy-regularized matrix game in Section~\ref{sec: entropy matrix games}. In particular, we let \(\lambda=0.1\), \(n=4\), \(m=12\), and \(b_i=\begin{bmatrix} 2 & \pi & \pi \end{bmatrix}^\top\) for all \(i=1, 2, 3, 4\). Here the elements in \(b_i\) denote the length of each candidate path, regardless of other players' strategies; if no other player exists, then action one (beeline path) is the optimal shortest path. If \(C=0_{12\times 12}\), one can verify---by solving an instance of optimization~\eqref{opt: nonlinear ls}---that the quantal response equilibrium in Definition~\ref{def: entropy Nash}---assuming \(\lambda\) is sufficiently small---is approximately 
\begin{equation}\label{eqn: routing initial Nash}
     x_i=\begin{bmatrix}
    1 & 0 & 0 
    \end{bmatrix}, \enskip i=1, 2, 3, 4.
\end{equation}
In other words, all players tend to choose the beeline path since it has the minimum length. However, this  causes collisions among the rovers at coordinate \((0, 0)\).

By adjusting the value of matrix \(C\), we aim to change the equilibrium above to the following
\begin{equation}\label{eqn: routing final Nash}
     x_i^\star=\begin{bmatrix}
    0 & 0 & 1 
    \end{bmatrix}, \enskip i=1, 2, 3, 4.
\end{equation}
In other words, we want all players to choose the counterclockwise semicircle path. See Fig.~\ref{fig: routing} for an illustration and \url{https://www.youtube.com/watch?v=EvtPp_DWqgU} for an animation. 

We note that one can change the equilibrium from \eqref{eqn: routing initial Nash} to \eqref{eqn: routing final Nash} by simply letting \(C=0_{m\times m}\) and modifying the value of \(b\). However, such a choice of parameter implies that all rovers will voluntarily choose a longer path regardless of other rovers' strategies, which has no meaningful interpretation in path planning.

\begin{figure}[!ht]
\centering
\includegraphics[trim=0.1cm 0.1cm 0.1cm 0.1cm,width=0.3\textwidth]{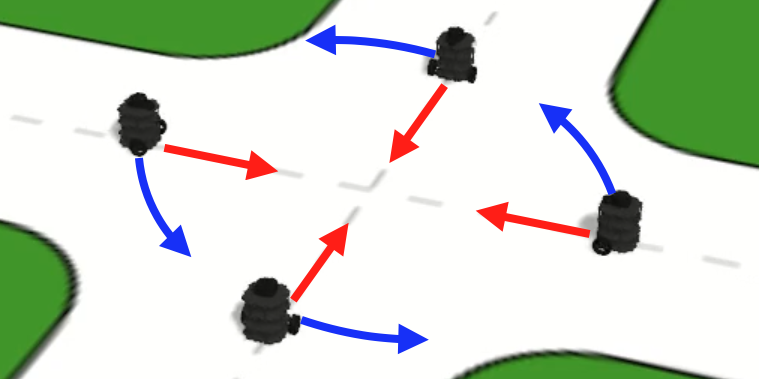}
\caption{An illustration of the equilibrium strategies in \eqref{eqn: routing initial Nash} (red arrows) and \eqref{eqn: routing final Nash} (blue arrows). }
  \label{fig: routing}
\end{figure}

Since the equilibrium in \eqref{eqn: routing final Nash} is of the form in Section~\ref{subsec: sdp}, we can compute matrix \(C\) using either the semidefinite program~\eqref{opt: sdp} or the bilevel optimization~\eqref{opt: bilevel}; in the latter case, we choose the performance function to be the KL-divergence in \eqref{eqn: KL div}. 

We solve the semidefinite program~\eqref{opt: sdp} using the off-the-shelf solver, and the bilevel optimization~\eqref{opt: bilevel} using Algorithm~\ref{alg: proj grad}. Fig.~\ref{fig: tradeoff} shows the trade-off between \(D_{KL}(x, x^\star)\)---which measures the distance between the equilibrium \(x\) that corresponds to matrix \(C\) and the desired equilibrium \(x^\star\)---and \(\norm{C}_F\) of the computed matrix \(C\) when tuning the parameter in \eqref{opt: sdp} and \eqref{opt: bilevel}. These results confirm that both the semidefinite program approach and the bilevel optimization approach apply to the cases where the desired Nash is pure and known explicitly. Furthermore, both approaches require careful tuning of algorithmic parameters to achieve a preferred trade-off between \(D_{KL}(x, x^\star)\) and \(\norm{C}_F\).

\begin{figure}[!ht]
\centering
  \begin{subfigure}{0.48\columnwidth}
  \includegraphics[trim=0.1cm 0.1cm 0.1cm 0.1cm,width=\textwidth]{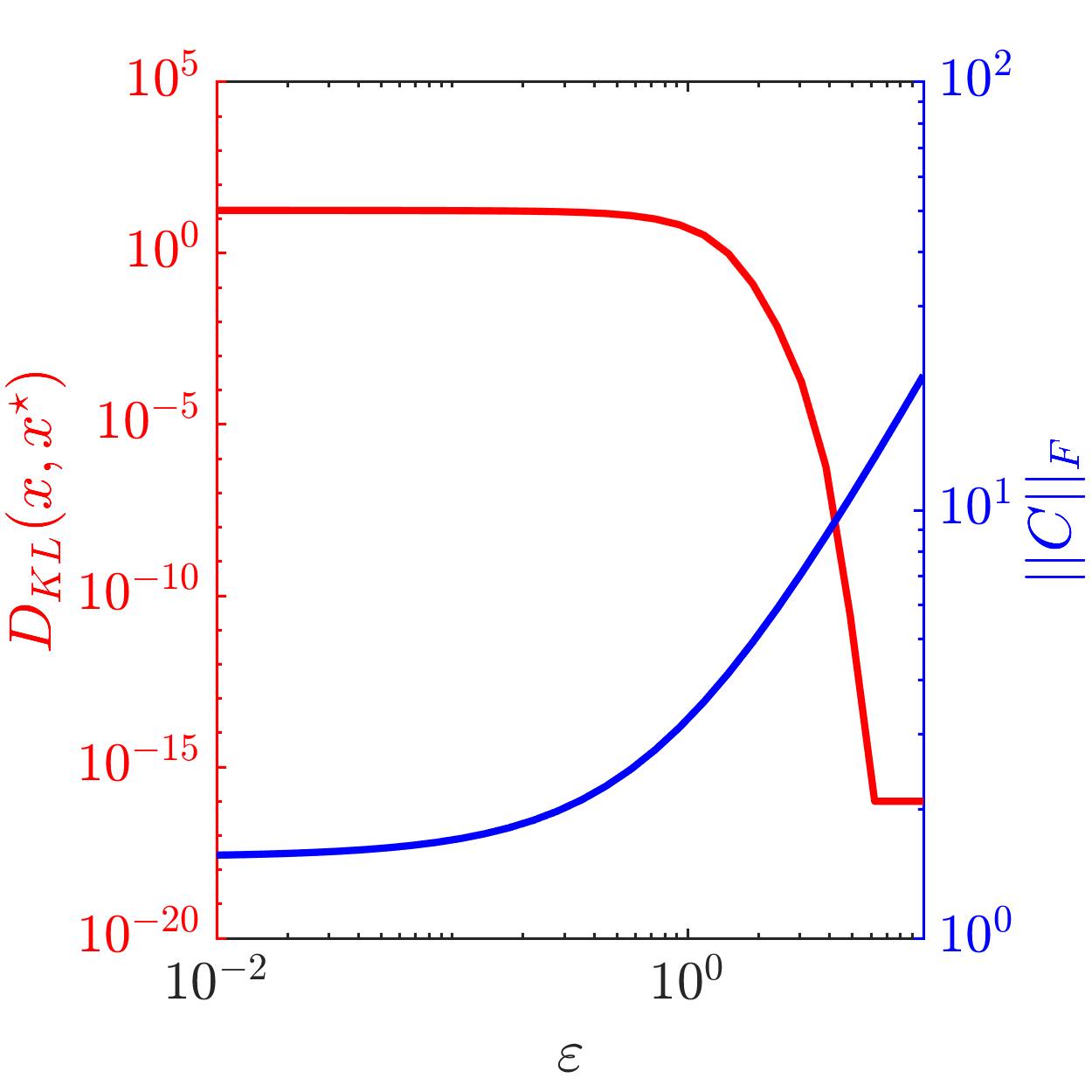}
  \caption{Semidefinite program~\eqref{opt: sdp}.}\label{fig: epsilon}
  \end{subfigure}
  \hfill
  \begin{subfigure}{0.48\columnwidth}
  \includegraphics[trim=0.1cm 0.1cm 0.1cm 0.1cm,width=\textwidth]{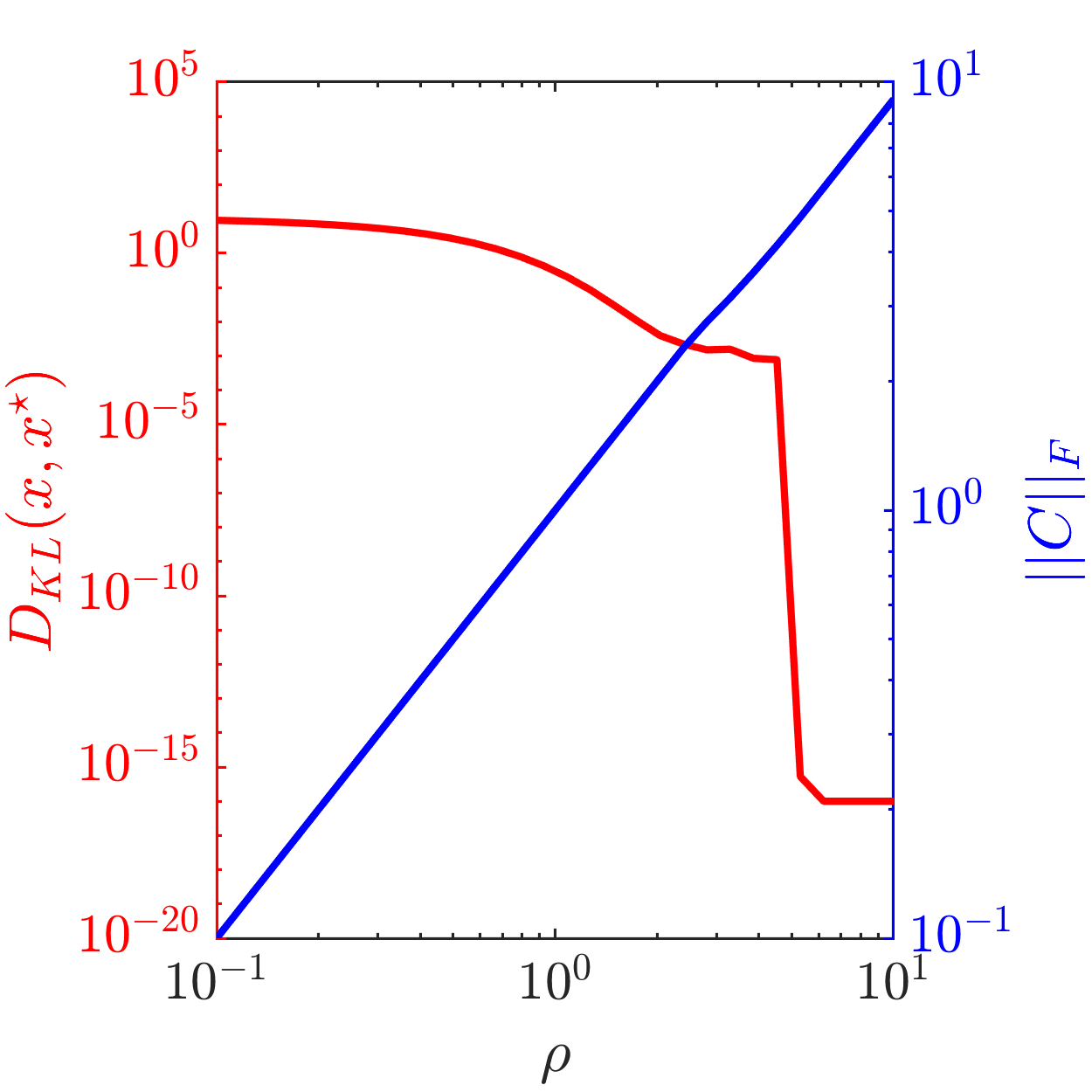}
  \caption{Bilevel optimization~\eqref{opt: bilevel}.} \label{fig: gamma}
  \end{subfigure} 
  \caption{The trade-off between \(D_{KL}(x, x^\star)\) and \(\norm{C}_F\) when tuning the parameter \(\epsilon\) in the semidefinite program~\eqref{opt: sdp} and the parameter \(\rho\) in bilevel optimization~\eqref{opt: bilevel}.  }
  \label{fig: tradeoff}
\end{figure}

\subsection{Encouraging fair resource allocation}

We now consider a case where the desired equilibrium is not of the explicit form in Section~\ref{subsec: sdp}. Instead, we only have access to a performance function that implicitly describes the desired equilibrium. To this end, we consider the following three-player game. Each player is a delivery drone company that provides package-delivery service, located in the southwest, southeast, and east area of Austin, respectively. Each strategy demotes the distribution of service allocated to the nine areas of Austin; we assume all three companies have the same amount of service to allocate. For each company, within its home area (where it is located), the operating cost of delivery service is one unit; outside the home area, the operating cost increases by \(50\%\) in an area adjacent to the home area, and \(80\%\) otherwise. See Fig.~\ref{fig: austin} for an illustration\footnote{Picture credit: \url{https://en.wikipedia.org/wiki/List_of_Austin_neighborhoods}}. We model the joint decision of the three companies using the matrix game in Section~\ref{sec: entropy matrix games}, where \(n=3\), \(m_i=9\) for \(i=1, 2, 3\), and \(m=27\); we set \(\lambda=0.1\) and vector \(b\) according to the aforementioned operating cost. 

\begin{figure}[!ht]
    \centering
    \input{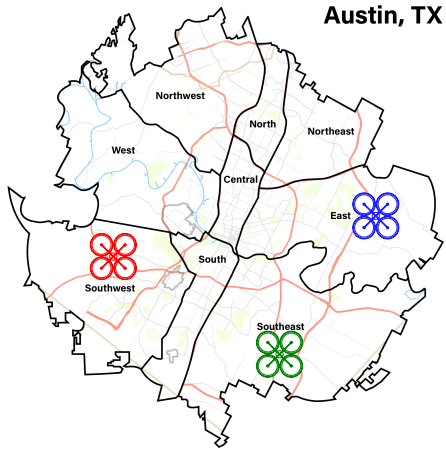}
    \caption{Three drone delivery companies (red, green, and blue) located in different areas in the city of Austin.}
    \label{fig: austin}
\end{figure}

If all companies consider only the operating cost, they will only allocate services to their respective home area. We aim to infer the value of matrix \(C\) using Algorithm~\ref{alg: proj grad} that encourages a fair allocation to other areas. In particular, we choose the performance function as follows:
\begin{equation}\label{eqn: alpha fairness}
    \psi(x) =\mathbf{1}_9^\top (x_1+x_2+x_3)^{-1}, 
\end{equation}
where vector \((x_1+x_2+x_3)^{-1}\) denotes the elementwise reciprocal of vector \(x_1+x_2+x_3\). Function \(\psi(x)\) is based on the the \emph{potential delay function} from the resource allocation literature; the latter is a special case of the more general \(\alpha\)-fairness function \cite[Sec. 2.4]{shakkottai2008network}. Here function \(\psi(x)\) measures the overall fairness of the delivery service allocation. Here we implicitly assume that the demand for delivery services is much higher than the supply, and we aim to allocate all supply. Such an assumption is common in the resource allocation literature \cite{shakkottai2008network}. The competition is among different suppliers (companies), not between supply and demand. 

We compute the cost matrix using Algorithm~\ref{alg: proj grad} and illustrate the percentages of the delivery service allocated to each area at the equilibrium in Fig.~\ref{fig: bar}. The results show that when \(\rho\approx 0\), all the drone fleets will almost only serve their respective home areas. As we increase the value of \(\rho\), the computed matrix encourages a more fair joint strategy where all nine areas receive an almost equal amount of service. 

\begin{figure}[!ht]
    \centering
    \includegraphics[trim=0.1cm 0.1cm 0.1cm 0.1cm,width=\linewidth]{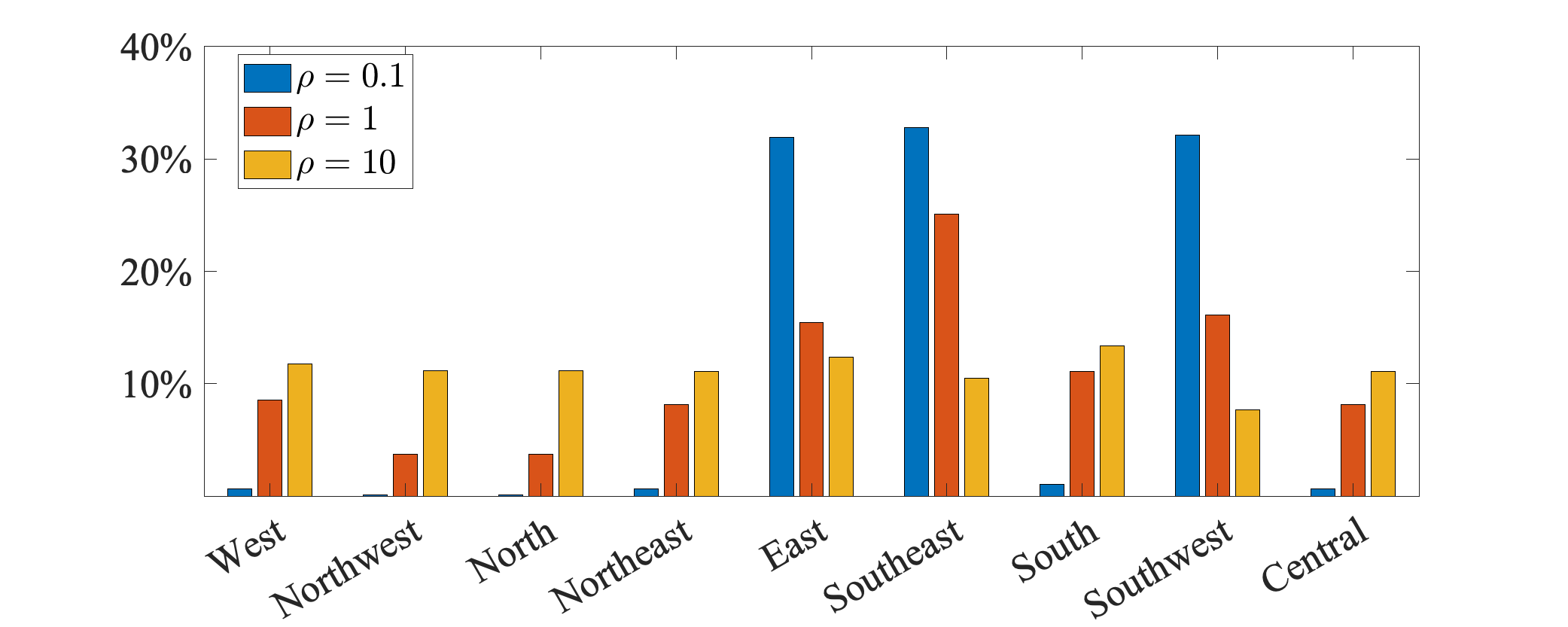}
    \caption{The percentages of the total amount of delivery service allocated to each area at the equilibrium computed by Algorithm~\ref{alg: proj grad}. }
    \label{fig: bar}
\end{figure}

\section{Conclusion}
\label{sec: conclusion}

We study the inverse game problem in the context of multiplayer matrix games. We develop efficient numerical methods to compute the cost matrices that ensure a unique quantal response equilibrium.

However, the current work only provides a preliminary proof of concept with limited applications. For example, it requires an exhaustive enumeration of all actions, which is computationally unscalable and makes the cost inference of a pure target equilibrium trivial: one can simply assign the lowest cost to the target actions. Furthermore, it only provides the cost functions that induce one desired equilibrium, rather than multiple equilibria with common desired properties. We aim to address these limitations by considering games with more complicated decision-making models. We will also consider  simultaneously optimizing multiple equilibria using robust optimization.



\section*{ACKNOWLEDGMENT} The authors would like to thank Yigit E. Bayiz, Shenghui Chen, Benjamin Grimmer, Dayou Luo, Shahriar Talebi, and the anonymous editor and reviewers for their feedback.

\bibliographystyle{IEEEtran}
\bibliography{IEEEabrv,reference}

\end{document}